\documentclass[10pt,journal,compsoc]{IEEEtran}

\usepackage{stmaryrd}
\usepackage{amsfonts}
\usepackage{mathrsfs,amsmath}
\usepackage{amssymb}
\usepackage{makeidx}
\usepackage{multirow}
\usepackage{algorithmic} 
\usepackage{bm}
\usepackage{setspace}
\usepackage{amsthm}
\usepackage{empheq}
\usepackage{booktabs}
\usepackage[ruled,linesnumbered,vlined]{algorithm2e}  

    \SetCommentSty{mycommfont}
\usepackage{color}
\usepackage{enumerate}
\usepackage[left=0.65in, right=0.65in, top=0.7in, bottom = 0.98in]{geometry}
\usepackage{makecell}
\usepackage{threeparttable}
\usepackage{booktabs}
\usepackage{tabularx}
\usepackage{array}
\usepackage[font=footnotesize]{caption}
\usepackage{tabularx,booktabs}
\usepackage{ifpdf}
 \ifpdf
 \else
 \fi
 
\usepackage{cite}
 
\ifCLASSINFOpdf
   \usepackage[pdftex]{graphicx}
\else
\fi
 
\usepackage{amsmath}
\usepackage[percent]{overpic}

\usepackage{array} 
\usepackage{tikz,pgfplots,filecontents}
\usetikzlibrary{spy} %
\pgfplotsset{width=8.6cm, height=6.5cm, compat=1.9}
 
\usepackage{url}
\usepackage{hyperref}

\makeatletter
\let\NAT@parse\undefined
\makeatother

\newtheorem{theorem}{Theorem}

\IEEEoverridecommandlockouts   

\allowdisplaybreaks

\pgfplotsset{every axis legend/.style={%
    cells={anchor=west},
    inner xsep=3pt,inner ysep=2pt,nodes={inner sep=0.8pt,text depth=0.15em},
    anchor=north east,%
    shape=rectangle,%
    fill=white,%
    draw=black,
    at={(0.98, 0.98)},
    font=\footnotesize,
    }
}

\pgfplotsset{every axis/.append style={line width=0.6pt,tick style={line width=0.8pt}}}

\begin{document}

\title{Multi-cell Content Caching: Optimization for Cost and Information Freshness}

\author{Zhanwei~Yu,~\IEEEmembership{Student Member,~IEEE,}
        Tao~Deng, Yi~Zhao, and~Di~Yuan,~\IEEEmembership{Senior~Member,~IEEE}
\IEEEcompsocitemizethanks{\IEEEcompsocthanksitem Z. Yu, Y. Zhao, and D. Yuan are with the Department of Information Technology, Uppsala University, 751 05 Uppsala, Sweden.\protect\\
	E-mail: \{zhanwei.yu; yi.zhao; di.yuan\}@it.uu.se
\IEEEcompsocthanksitem T. Deng is with School of Information Science and Technology, Southwest Jiaotong University, 610032 Sichuan, China.\protect\\
	E-mail: dengtaoswjtu@foxmail.com}
\thanks{Manuscript submitted May 27, 2022. This article was presented in part at the IEEE Wireless Communications and Networking Conference Workshops (WCNC WKSHPS) , April 2022 \cite{yu2022multi}.}}


\IEEEtitleabstractindextext{%
\begin{abstract}
In multi-access edge computing (MEC) systems, there are multiple local cache servers caching contents to satisfy the users' requests, instead of letting the users download via the remote cloud server. In this paper, a multi-cell content scheduling problem (MCSP) in MEC systems is considered. Taking into account jointly the freshness of the cached contents and the traffic data costs, we study how to schedule content updates along time in a multi-cell setting. Different from single-cell scenarios, a user may have multiple candidate local cache servers, and thus the caching decisions in all cells must be jointly optimized. We first prove that MCSP is NP-hard, then we formulate MCSP using integer linear programming, by which the optimal scheduling can be obtained for small-scale instances. For problem solving of large scenarios, via a mathematical reformulation, we derive a scalable optimization algorithm based on repeated column generation. Our performance evaluation shows the effectiveness of the proposed algorithm in comparison to an off-the-shelf commercial solver and a popularity-based caching.
\end{abstract}

}

\maketitle

\IEEEdisplaynontitleabstractindextext

%
\IEEEpeerreviewmaketitle

\IEEEraisesectionheading{\section{Introduction}\label{Sec:Intro}}

For streaming media applications in emerging scenarios, the existing communication systems face significant challenges in delivering ultra-high bandwidth, ultra-large storage, and ultra-low latency. Multi-access edge computing (MEC) is widely accepted as a key technology in this context \cite{8016573}, as it can provide computation-intensive and caching-intensive infrastructure. 

This paper focuses on content caching in multi-cell MEC systems. In each cell, the local cache server downloads contents from the remote cloud server via a backhaul link to provide services for users. Moreover, there are overlapping areas among the cells, thus some users potentially can obtain cached contents from more than one cell. For user experience in terms of the freshness of contents, cache servers need to update cached contents frequently. We use the notion of age of information (AoI), introduced in \cite{kaul2012real}, as the metric of information freshness. We address optimal scheduling of cache updates with respect to the trade-off between the freshness of contents and the cost due to downloading contents.

We consider a finite and time-slotted scheduling horizon. There are two types of limits in cache servers. Downloading and caching in a time slot are subject to backhaul and cache capacities, respectively. A request can be satisfied by a cache if the demanded content is cached, otherwise the content will be downloaded from the cloud to the user, resulting in a higher cost. We use a cost function addressing the trade-off between freshness and download cost. To minimize the overall cost, the updating and caching decisions of cells are coupled such that we need to coordinate the cells instead of optimizing them in isolation. 

Our work consists of the following contributions for the outlined multi-cell content scheduling problem (MCSP).

\begin{itemize}
	\item We provide a complexity analysis for MCSP. Specifically, even for the special case of uniform content size and single time slot, the problem is NP-hard, based on a polynomial reduction from Monotone 3-Satisfiability (M3SAT). 
	\item We provide an integer linear programming (ILP) formulation for MCSP. Using commercial off-the-shelf optimization solvers, the global optimum can be found within an acceptable amount of time for small-scale instances. 
	\item We derive a mathematical reformulation that enables a scalable solution approach. Specifically, in the reformulation, sequences are introduced to represent updating and caching decisions of contents in cache servers over the time slots. 
	\item We propose a column generation algorithm that can efficiently solve the linear programming (LP) version of the reformulation. The algorithm keeps augmenting a small subset of sequences until the optimality condition is satisfied. To recover integrality, we provide a rounding strategy based on disjunctions. The column generation and the rounding strategy are alternately applied to deliver a problem solution. 
	\item Performance evaluation demonstrates the superiority of our proposed algorithm in comparison to a popularity-based caching and a commercial solver. Moreover, the algorithm provides an optimality bound, so the performance can be gauged even if the global optimum is out of reach. Using the bound, our results show the algorithm consistently yields near-to-optimal solutions that are at most $2.4\%$ from the global optimum for 3-cell and 7-cell instances.
\end{itemize}

\section{Related Work}\label{Sec:Works}

Edge content caching has been for example addressed in \cite{bastug2014living, liang2018enhancing, 9296970, vu2018edge}. Leveraging social networks and device-to-device communications, the authors of \cite{bastug2014living} propose a proactive caching paradigm that can improve the spectral efficiency of backhaul as well as the ratio of satisfied users. In \cite{liang2018enhancing}, caching is applied to software-defined mobile networks. The work in \cite{9296970} proposes a distributed delay optimization algorithm for scenarios where multiple unmanned aerial vehicles are equipped with cache to serve users. The authors of \cite{vu2018edge} investigate multi-layer caching where both base station and users can cache the contents. 

The works in \cite{tan2012optimal, deng2018cost, wen2017cache, chen2019collaborative, banerjee2018greedy, Poularakis2020Service, 9139263, 8891760} have studied cache content placement strategies based on popularity. In \cite{tan2012optimal}, the authors provide an optimal content place strategy for large-scale peer-to-peer systems to maximize the utilization of the uplink bandwidth resources. The authors of \cite{deng2018cost} consider the impact of user mobility in device-to-device networks, and propose a mobility-aware multi-user algorithm to minimize the expected average cost per user. The work in \cite{wen2017cache} proposes a strategy for the popular random-heterogeneous-cellular-network model based on Poisson point processes to maximize the hit probability over content placement probabilities. By utilizing network coding, content files are partitioned into multiple coded segments, and the contents can then be stored with distributed caching. The work in \cite{chen2019collaborative} provides a strategy for the distribution of the coded segments among multiple cache servers to reduce the content downloading time. The content placement strategies in \cite{tan2012optimal, deng2018cost, wen2017cache, chen2019collaborative} primarily focus on pushing popular content to the network edge. The strategy proposed in \cite{banerjee2018greedy} jointly optimizes the caches in edge and core networks to maximize the hit ratio. The work in \cite{Poularakis2020Service} studies the joint optimization of cache placement and request routing in a multi-cell MEC system. In \cite{9139263} and \cite{8891760}, reinforcement learning approaches are applied to design cache content placement strategies.

The AoI concept as a metric for information freshness is introduced in \cite{kaul2012real}. The work in \cite{sun2017update} reveals that AoI is significantly different from delay, and the aspect of AoI in many application scenarios has been addressed in \cite{8845089, farazi2019fundamental, 8406914, 8845254, liu2021aion}. In \cite{8845089}, the authors propose an algorithm for optimal link activation to minimize energy consumption subject to AoI constraints in wireless networks with interference. The work in \cite{farazi2019fundamental} studies the AoI of status updates in a general multi-source multi-hop wireless network with the presence of interference, and the emphasis is on performance bounds. The work in \cite{8406914}  considers the joint optimization of node assignment and camera transmission scheduling to minimize AoI in wireless camera networks with fog computing. In \cite{8845254}, the authors study AoI with respect to orthogonal multiple access (OMA) and non-orthogonal multiple access (NOMA), and show that NOMA is not always better than OMA in terms of average AoI. The work in \cite{liu2021aion} investigates the bandwidth minimization problem under AoI constraints in bandwidth-limited Internet of things.
 
Some works have provided caching strategies accounted for AoI. The authors of \cite{kam2017information} and \cite{ahani2020accounting} provide strategies based on both the popularity and the freshness of contents. In \cite{9387141}, a freshness-aware update scheme is proposed for the trade-off between the service delay and the freshness of contents. In \cite{ahani2020optimal}, the authors optimize content scheduling to address AoI and information absence. In \cite{xu2021optimal}, the authors provide a model-free reinforcement learning algorithm to minimize the long-term average cost subject to both AoI and energy consumption constraints. 

The scenarios in \cite{ahani2020accounting, 9387141, ahani2020optimal, xu2021optimal} are different from ours as our paper considers multi-cell caching. To the best of our knowledge, the work in \cite{9199125} is closest to ours as it considers a heterogeneous network where multiple cache servers are placed in cells to serve users. However, there are significant differences. Our system model has overlapping areas among the cells, and the cells are coupled due to such a topology. However, in \cite{9199125}, the cells are disjoint but coupled due to the shared update frequency. In addition, the authors of \cite{9199125} consider minimizing the average AoI of delivered contents while we jointly minimize the AoI and cost. 

\section{System Scenario and Problem Formulation}\label{Sec:SSPF}

\subsection{System Scenario}\label{Sec:SSPF:Subsec:SS}

The MEC scenario consists of a cloud server and $H$ cells as illustrated in Fig. \ref{fig:system_model}. There is a cache server in each cell. To model realistic network topologies, there are overlapping areas among the cells. In this paper, the terms cell and cache server are used interchangeably. Denote by $\mathcal{H}=\{1,2,\dots,H\}$ the set of cache servers, and $C_h$ and $G_h$ the cache capacity and the backhaul capacity of cache server $h$, respectively. The cache servers download contents from the cloud server via backhaul links. The contents subject to caching form a set $\mathcal{I} = \{1, 2, \dots, I\}$, and $s_i$ is the size of content $i$. The caching system targets satisfying content requests by cached contents, to reduce the load of core network and the response latency. 

\input{fig-system_model}

A scheduling horizon consisting of a set of time slots $\mathcal{T} = \{1,2,\dots,T\}$ is considered in this paper\footnote{The time slot duration in our case is at the magnitude of minutes or even an hour \cite{9130754}.}. Denote by $\mathcal{R} = \{1,2,\dots, R\}$ the set of the requests\footnote{The requests can be known via for example a prediction model from the work in \cite{zhang2018using}.}. The requests can be divided into two types:
\begin{enumerate}
	\item Multiple-choice requests (MCRs): The requests are from the users located in the overlapping areas, thus there are multiple candidate cache servers. 
	\item Single-choice requests (SCRs): These are simply the rest of requests that can be satisfied by only one cache server.
\end{enumerate}
We use $\mathcal{R}^{\text{m}}$ and $\mathcal{R}^{\text{s}}$ to represent the sets of MCRs and SCRs, respectively. Obviously, $\mathcal{R} = \mathcal{R}^{\text{m}} \cup \mathcal{R}^{\text{s}}$. Each request $r$ consists of a tuple $(i_r, o_r, d_r, \mathcal{H}_r)$ where: 
\begin{itemize}
	\item $i_r$ is the content of requested by $r$. 
	\item $o_r$ is the time slot of request $r$.
	\item $d_r$ is the time slot that is the deadline of satisfying request $r$.
	\item $\mathcal{H}_r$ is the set of candidate cache server(s) for request $r$.
\end{itemize}

Our problem consists in making scheduling decisions of content caching and updating over the time slots. We use AoI to model the freshness of contents. A content that is just downloaded from the cloud to a cache server has AoI zero, and the AoI increases by one for every time slot as long as the content is still cached but not updated. Moreover, for a request, if the requested content is absent from the candidate cache server(s) by the deadline, the content will be downloaded from the cloud, and the AoI is zero in this case.

\subsection{Optimization Variables}\label{Sec:SSPF:Subsec:OV}

Denote by $x_{hita}$ a binary variable that is one if and only if content $i$ is cached with AoI $a$ in cache server $h$ in time slot $t$. Based on the definition of AoI, the $x$-variables are constrained by the following, that states how AoI evolves over time.
\begin{align}
	& x_{hi(t-1)(a-1)} \geq x_{hita}, \notag\\
	&\qquad\ \ \forall h \in \mathcal{H}, i \in \mathcal{I}, t \in \mathcal{T} \setminus  \{1\}, a \in \{1,2,\dots,t-1\}. \label{ILP_c1}
\end{align}
We use binary variable $z_{hit}$ that equals one if cache server $h$ caches content $i$ in time slot $t$. Then 
\begin{equation}\label{ILP_c3}
	\sum_{a=0}^{t-1}x_{hita}=z_{hit}, \forall h \in \mathcal{H}, i \in \mathcal{I}, t \in \mathcal{T}.
\end{equation}
Binary variable $y_{rha}$ is one if request $r$ is satisfied with AoI $a$ by cache server $h$. A request needs to be responded to at most once by the cache servers, hence we have
\begin{equation}\label{ILP_c4}
	\sum_{h \in \mathcal{H}_r} \sum_{a = 0}^{d_r - 1} y_{rha} \leq 1, \forall r \in \mathcal{R}.\\
\end{equation}
Furthermore, request $r$ can be satisfied with AoI $a$ by cache server $h$ only if there exists content $i_r$ whose AoI is $a$ in time slots $t \in \{o_r, \dots, d_r\}$ of cache server $h$, i.e.,
\begin{equation}\label{ILP_c5}
	y_{rha} \leq \sum^{d_r}_{t = o_r} x_{hi_rta}, \forall r \in \mathcal{R}, h \in \mathcal{H}_r, a \in \{0,1,\dots,d_r-1\}.
\end{equation}

\subsection{Cost Model}\label{Sec:SSPF:Subsec:CM}

From a pure AoI standpoint, a naive and extreme setting is that all the requests are met by the cloud server. However, the costs due to downloading would be very large. We consider the trade-off between the costs due to downloading from the cloud (to cache servers and directly by users for uncached contents) and the performance in AoI. In this paper, we use a joint cost function to address the trade-off, and our cost function consists of three components as follows.

\subsubsection{AoI Cost}
Denote by $f_i(a)$ $(a \in  \mathbb{N})$ a monotonically increasing function to quantify the impact of AoI of content $i$. A specific candidate function $f_i(a)$ can be, for example, $f_i(a) = {\rm e}^{a}$ that is the reciprocal of a utility function with respect to AoI in \cite{5062058}. The total AoI cost for all the requests is expressed as
\begin{align}\label{AAA}
	&A(\boldsymbol{y}) = \notag\\
	&\sum_{r\in \mathcal{R}} \left[ \sum_{h\in \mathcal{H}_r} \sum_{a = 0}^{d_r-1} f_{i_r}(a)y_{rha} \!+\! f_{i_r}(0)\left(1 \!-\! \sum_{h\in \mathcal{H}_r} \sum_{a = 0}^{d_r-1}y_{rha}\right)\right].
\end{align}
The first term in the square brackets is the AoI cost if a request is satisfied by a cache server, and the second term is for the case where the content with AoI zero is directly downloaded from the cloud server. Exactly one of the two cases will occur due to \eqref{ILP_c4}.

\subsubsection{Download Cost}

Denote by $\alpha$ the unit download cost when users download from the cloud server. The total download cost is thus
\begin{equation}\label{DDD}
	D(\boldsymbol{y}) = \alpha \sum_{r\in \mathcal{R}} s_{i_r} (1 -  \sum_{h \in \mathcal{H}_r} \sum_{a = 0}^{d_r-1} y_{rha}),
\end{equation}
where $s_{i_r}$ is the size of content $i_r$. Obviously, for request $r$, a download cost of $\alpha s_{i_r}$ occurs if all the $y$-variables of request $r$ are zero, i.e., the download cost will occur only if the request is not satisfied by cache servers. 

\subsubsection{Update Cost}

We use $\beta$ $(\beta < \alpha)$ to represent the unit cost when cache servers download contents via the backhaul from the cloud server. Because $x_{hit0}$ represents that cache server $h$ downloads content $i$ in time slot $t$, the update cost can be expressed by
\begin{equation}
	U(\boldsymbol{x}) = \beta \sum_{h\in \mathcal{H}} \sum_{i\in \mathcal{I}} \sum_{t\in \mathcal{T}} x_{hit0} s_i.
\end{equation}

\subsection{Problem Complexity}\label{Sec:SSPF:Subsec:HA}
Consider MCSP, subject to the cache capacity and backhaul capacity. We analyze its complexity in Theorem \ref{NPhard}, which implies that no low-complexity and exact algorithm can be expected for MCSP, unless $\mathcal{P} = \mathcal{NP}$.
\begin{theorem}\label{NPhard}
	MCSP is NP-hard even if the contents are of uniform size.
\end{theorem}
\begin{proof}
    The proof is established by a reduction from M3SAT that is NP-complete \cite{M3SAT}. Consider a formula of clauses $\mathcal{F} = \land _{k = 1}^{N} \mathcal{C}_k$, where each clause $\mathcal{C}_k$ has at most three negated or non-negated boolean variables. The M3SAT problem is to determine whether or not there exists an assignment of true/false values to the variables, such that formula $\mathcal{F}$ is true.
    
    Consider an instance of M3SAT with $N$ clauses and $M$ variables, we construct a corresponding MCSP instance with $M$ cache servers, $N$ primary requests, and $2M$ auxiliary requests. We first set some parameters of the network: 
    \begin{enumerate}
    	\item There are only two contents A and B of equal size. Denote the size by $s$.
    	\item There is only one time slot, therefore, the AoI is always zero. Also, the AoI costs are equal for the two contents, i.e., $f_{A}(0) = f_{B}(0) = f$.
    	\item Each cache server has two auxiliary requests that are SCRs for content A and B, respectively. 
    	\item A primary request might be SCR or MCR. If a primary is MCR, it has at most three candidate cache servers.  
    \end{enumerate}
    We define the primary requests according to the clauses in $\mathcal{F}$. For any clause, there is a primary request for content A if all the variables are negated in the clause. Conversely, there is a primary request for content B if all the variables are non-negated. For any primary request, its candidate cache server(s) will be set to correspond to the variable(s) in the clause. 
    
    
    By the construction, because $\alpha$ and $\beta$ are both positive, and $\alpha>\beta$, with the two auxiliary requests, each cache server must store content A or B, so the update cost is a constant of $M\beta s$. Moreover, exactly half of the auxiliary requests are satisfied by cached contents, and the other half must use the cloud, and, therefore, the sum of AoI and download costs with respect to auxiliary requests equals $2Mf + M\alpha s$. For the primary requests, the AoI cost is $Nf$, and zero download cost can be achieved if and only if all the primary requests are satisfied by the cache servers. Now, the question is then if the lower bound of $M\beta s + 2Mf + M\alpha s + Nf$ on the overall cost can be achieved. This question can be answered by solving the defined MCSP instance, and the solution gives the correct answer of M3SAT. Therefore, MCSP is strongly NP-hard.
\end{proof}

\subsection{Problem Formulation}\label{Sec:SSPF:Subsec:PF}
MCSP can be formulated by integer linear programming (ILP) as follows.
\begin{subequations}\label{formulation}
	\begin{align}
		\underset{\boldsymbol{x},\boldsymbol{y},\boldsymbol{z} \in \{0,1\}}\min 	& \ \ A(\boldsymbol{y}) + 	D(\boldsymbol{y}) + 	U(\boldsymbol{x})\label{ILP_obj}\\
		\textup{ s.t.}\ \ \ & \text{\eqref{ILP_c1}-\eqref{ILP_c5}}, \notag\\
		& \sum_{i\in \mathcal{I}}s_i z_{hit} \leq C_h, \forall t\in \mathcal{T},h\in \mathcal{H}\label{ILP_SC}\\
		& \sum_{i \in \mathcal{I}}s_i x_{hit0} \leq G_h, \forall t\in \mathcal{T},h\in \mathcal{H}\label{ILP_BC}
	\end{align}
\end{subequations}

The objective function \eqref{ILP_obj} is the overall cost. Constraints \eqref{ILP_c1}-\eqref{ILP_c5} state the relationship among $\boldsymbol{x}$, $\boldsymbol{y}$, and $\boldsymbol{z}$. Constraints \eqref{ILP_SC} and \eqref{ILP_BC} state the cache and backhaul capacity, respectively. The ILP model can be solved by solvers e.g., Gurobi \cite{GUROBI}, though they need significant computational time. For efficient solution of MCSP, we next provide an algorithm based on repeated column generation. 

\section{Algorithm Design} \label{Sec:AD}

\subsection{Reformulation of MCSP} \label{Sec:AD:Subsec:Rfm}
We first give a reformulation of MCSP. Let us consider any of the cache servers, and the caching and updating solution for a particular content on the server. This can be viewed along the time slots, namely that for a time slot, the content is either cached or not, represented by one or zero, respectively. In the former case, it is also significant to know if the content gets updated. Subsequently, we will use the term update to refer to downloading a content that is either added to the cache or refreshed with respect to AoI. A sequence that represents the solutions of one content, mathematically, is called a column. 

In other words, a column for any content in a cache server is a vector representing the caching and updating decisions of the content over all time slots. We denote by $\mathcal{L}$ the index set of all possible such vectors. Each column $\ell \in \mathcal{L}$ consists of $T$ tuples $(q_{t\ell}, p_{t\ell})$ over all the time slots, i.e., $\ell = \left[(q_{1\ell}, p_{1\ell}), (q_{2\ell}, p_{2\ell}), \dots, (q_{T\ell}, p_{T\ell})\right]^ \mathrm{T} $, where $q_{t\ell}$ and $p_{t\ell}$ are one if and only if the status is cached and updated in time slot $t$, respectively. An example of the columns is given in Fig. \ref{fig:column_example}. Because updating a content implies it is cached, a tuple $(q_{t\ell}, p_{t\ell})$ has three possible patterns: $(1, 1)$, $(1, 0)$, and $(0, 0)$. Therefore, the cardinality of set $\mathcal{L}$ is bounded by $3^T$.

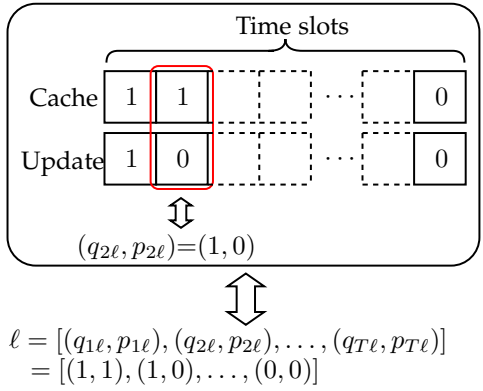
\begin{figure}[t]
    \begin{center}
        \tikzset{every picture/.style={line width=0.75pt}} 
        \begin{tikzpicture}[x=0.75pt,y=0.75pt,yscale=-1,xscale=1, scale = 0.85]
            
            \draw   (126.13,94.77) -- (156.75,94.77) -- (156.75,125.39) -- (126.13,125.39) -- cycle ;
            \draw   (156.75,94.77) -- (187.37,94.77) -- (187.37,125.39) -- (156.75,125.39) -- cycle ;
            \draw   (309.84,94.77) -- (340.45,94.77) -- (340.45,125.39) -- (309.84,125.39) -- cycle ;
            \draw   (126.13,131.51) -- (156.75,131.51) -- (156.75,162.13) -- (126.13,162.13) -- cycle ;
            \draw   (156.75,131.51) -- (187.37,131.51) -- (187.37,162.13) -- (156.75,162.13) -- cycle ;
            \draw   (309.84,131.51) -- (340.45,131.51) -- (340.45,162.13) -- (309.84,162.13) -- cycle ;
            \draw  [dash pattern={on 2pt off 2pt}]  (248.6,94.77) -- (248.6,125.39) ;
            \draw  [dash pattern={on 2pt off 2pt}]  (217.98,94.77) -- (248.6,94.77) ;
            \draw  [dash pattern={on 2pt off 2pt}]  (217.98,125.39) -- (248.6,125.39) ;
            \draw  [dash pattern={on 2pt off 2pt}]  (279.22,94.77) -- (279.22,125.39) ;
            \draw  [dash pattern={on 2pt off 2pt}]  (279.22,94.77) -- (309.84,94.77) ;
            \draw  [dash pattern={on 2pt off 2pt}]  (279.22,125.39) -- (309.84,125.39) ;
            \draw  [dash pattern={on 2pt off 2pt}]  (248.6,131.51) -- (248.6,162.13) ;
            \draw  [dash pattern={on 2pt off 2pt}]  (217.98,131.51) -- (248.6,131.51) ;
            \draw  [dash pattern={on 2pt off 2pt}]  (217.98,162.13) -- (248.6,162.13) ;
            \draw  [dash pattern={on 2pt off 2pt}]  (279.22,131.51) -- (279.22,162.13) ;
            \draw  [dash pattern={on 2pt off 2pt}]  (279.22,131.51) -- (309.84,131.51) ;
            \draw  [dash pattern={on 2pt off 2pt}]  (279.22,162.13) -- (309.84,162.13) ;
            \draw  [color={rgb, 255:red, 255; green, 0; blue, 0 }  ,draw opacity=1 ] (153.5,96.46) .. controls (153.5,93.86) and (155.61,91.75) .. (158.21,91.75) -- (185.79,91.75) .. controls (188.39,91.75) and (190.5,93.86) .. (190.5,96.46) -- (190.5,160.04) .. controls (190.5,162.64) and (188.39,164.75) .. (185.79,164.75) -- (158.21,164.75) .. controls (155.61,164.75) and (153.5,162.64) .. (153.5,160.04) -- cycle ;
            \draw   (339.7,89.7) .. controls (339.7,85.03) and (337.37,82.7) .. (332.7,82.7) -- (243.5,82.7) .. controls (236.83,82.7) and (233.5,80.37) .. (233.5,75.7) .. controls (233.5,80.37) and (230.17,82.7) .. (223.5,82.7)(226.5,82.7) -- (134.3,82.7) .. controls (129.63,82.7) and (127.3,85.03) .. (127.3,89.7) ;
            \draw   (210,214.95) -- (220.9,222.56) -- (215.45,222.56) -- (215.45,237.79) -- (220.9,237.79) -- (210,245.4) -- (199.1,237.79) -- (204.55,237.79) -- (204.55,222.56) -- (199.1,222.56) -- cycle ;
            \draw   (68.6,71.51) .. controls (68.6,62.26) and (76.11,54.75) .. (85.36,54.75) -- (331.84,54.75) .. controls (341.09,54.75) and (348.6,62.26) .. (348.6,71.51) -- (348.6,193.49) .. controls (348.6,202.74) and (341.09,210.25) .. (331.84,210.25) -- (85.36,210.25) .. controls (76.11,210.25) and (68.6,202.74) .. (68.6,193.49) -- cycle ;
            \draw   (171.52,170.15) -- (177.84,174.56) -- (174.68,174.56) -- (174.68,183.39) -- (177.84,183.39) -- (171.52,187.8) -- (165.2,183.39) -- (168.36,183.39) -- (168.36,174.56) -- (165.2,174.56) -- cycle ;
            \draw  [dash pattern={on 2pt off 2pt}]  (217.98,94.77) -- (217.98,125.39) ;
            \draw  [dash pattern={on 2pt off 2pt}]  (187.37,94.77) -- (217.98,94.77) ;
            \draw  [dash pattern={on 2pt off 2pt}]  (187.37,125.39) -- (217.98,125.39) ;
            \draw  [dash pattern={on 2pt off 2pt}]  (187.37,131.51) -- (217.98,131.51) ;
            \draw  [dash pattern={on 2pt off 2pt}]  (187.37,162.13) -- (217.98,162.13) ;
            \draw  [dash pattern={on 2pt off 2pt}]  (217.98,131.51) -- (217.98,162.13) ;
            
            \draw (60, 244.5) node [anchor=north west][inner sep=0.75pt]    {$ \begin{array}{l}
            \ell =[( q_{1\ell } ,p_{1\ell }) ,\textcolor[rgb]{0,0,0}{(}\textcolor[rgb]{0,0,0}{q}\textcolor[rgb]{0,0,0}{_{2\ell }}\textcolor[rgb]{0,0,0}{,p}\textcolor[rgb]{0,0,0}{_{2\ell }}\textcolor[rgb]{0,0,0}{)} , \dots,( q_{T\ell } ,p_{T\ell })]\\
            \ \ \ =[( 1,1) ,\textcolor[rgb]{0,0,0}{( 1,0}\textcolor[rgb]{0,0,0}{)}\textcolor[rgb]{0,0,0}{,} \dots,( 0,0)]
            \end{array}$};
            \draw (202.22, 60) node [anchor=north west][inner sep=0.75pt]   [align=left] {Time slots};
            \draw (136.64, 103) node [anchor=north west][inner sep=0.75pt]    {$1$};
            \draw (167.08, 103) node [anchor=north west][inner sep=0.75pt]    {$1$};
            \draw (320.51, 103) node [anchor=north west][inner sep=0.75pt]    {$0$};
            \draw (136.64, 140) node [anchor=north west][inner sep=0.75pt]    {$1$};
            \draw (167.08, 140) node [anchor=north west][inner sep=0.75pt]    {$0$};
            \draw (320.51, 140) node [anchor=north west][inner sep=0.75pt]    {$0$};
            \draw (80, 103) node [anchor=north west][inner sep=0.75pt]   [align=left] {Cache};
            \draw (75, 140) node [anchor=north west][inner sep=0.75pt]   [align=left] {Update};
            \draw (255, 105) node [anchor=north west][inner sep=0.75pt]    {$\cdots $};
            \draw (255, 142) node [anchor=north west][inner sep=0.75pt]    {$\cdots $};
            \draw (108, 190) node [anchor=north west][inner sep=0.75pt]    {$\textcolor[rgb]{0,0,0}{(}\textcolor[rgb]{0,0,0}{q}\textcolor[rgb]{0,0,0}{_{2\ell }}\textcolor[rgb]{0,0,0}{,p}\textcolor[rgb]{0,0,0}{_{2\ell }}\textcolor[rgb]{0,0,0}{)}\textcolor[rgb]{0,0,0}{=}\textcolor[rgb]{0,0,0}{(}\textcolor[rgb]{0,0,0}{1,0}\textcolor[rgb]{0,0,0}{)}$};
        \end{tikzpicture}
    \end{center}
    \caption{An example of columns and the corresponding caching and updating decisions. The associated AoI values are directly induced from the caching and updating decisions.}\label{fig:column_example}
\end{figure}

We use a binary variable $\chi_{hi\ell}$ for column $\ell$, content $i$, and cache server $h$. This variable is one if and only if column $\ell$ is selected for content $i$ and cache server $h$. Because a column $\ell$ fully contains the caching and updating decisions, the associated AoI values are known for any given column. 

Hence, for any given column $\ell$ for content $i$ and cache server $h$, part of the overall cost of MCSP is determined. This includes
\begin{enumerate}
	\item The update cost of cache server $h$ for content $i$.
	\item The AoI and download costs of the server's SCRs. This set of the SCRs can be expressed by $\mathcal{R}^{\text{s}}_{hi} = \{r \in \mathcal{R}^{\text{s}}| (\mathcal{H}_r = \{h\}) \wedge (i_r = i)\}$.
\end{enumerate}
We use $S_{hi\ell}$ to denote the sum of the above costs.

The unknown cost component is related to the $y$-variables, that is, the AoI and download costs of the MCRs. We use $A^{\text{m}}(\boldsymbol{y})$ and $D^{\text{m}}(\boldsymbol{y})$ to represent them, and they can be expressed by replacing $\mathcal{R}$ with $\mathcal{R}^{\text{m}}$ in \eqref{AAA} and \eqref{DDD}, respectively.

Accordingly, MCSP can be reformulated as follows.
\begin{subequations}\label{MP}
	\begin{align}
		\underset{\boldsymbol{\chi},\boldsymbol{y}\in \{0,1\}}\min & \ \sum_{h\in \mathcal{H}} \sum_{i \in \mathcal{I}} \sum_{\ell \in \mathcal{L}} S_{hi\ell} \chi_{hi\ell} + A^{\text{m}}(\boldsymbol{y}) + D^{\text{m}}(\boldsymbol{y}) \\
		\text{s.t.} \ \ \ & \text{\eqref{ILP_c4}}, \notag\\
		& y_{rha} \leq \sum_{h \in \mathcal{H}_r}  \sum_{i \in \mathcal{I}} \sum_{\ell \in \mathcal{L}} B^{\ell}_{rha}\chi_{hi\ell}, \notag\\
		& \qquad \forall r \in \mathcal{R}^{\text{m}}, h \in \mathcal{H}_r, a \in \{0,1,\dots,t_r-1\} \label{MP_y}\\
		& \sum_{i \in \mathcal{I}}\sum_{\ell \in \mathcal{L}}  s_i q_{t\ell} \chi_{hi\ell} \leq C_h, \forall h\in \mathcal{H}, t\in \mathcal{T}\label{MP_SC}\\
		& \sum_{i \in \mathcal{I}}\sum_{\ell \in \mathcal{L}}  s_i p_{t\ell} \chi_{hi\ell} \! \leq \! G_h, \forall h\in \mathcal{H}, t\in \mathcal{T}\label{MP_BC}\\
		& \sum_{\ell \in \mathcal{L}} \chi_{hi\ell} = 1, \forall h \in \mathcal{H}, i \in \mathcal{I}\label{MP_one}
	\end{align}
\end{subequations}
Using the notion of column, we do not need the $x$-variables and the $z$-variables any more, thus constraints \eqref{ILP_c1} and \eqref{ILP_c3} are dropped. Constraint \eqref{MP_y} is derived from constraint \eqref{ILP_c5}, where $B^{\ell}_{rha}$ is one if and only if column $\ell$ for content $i$ and cache server $h$ can satisfy request $r$ with AoI $a$. As the columns include all the scheduling information, the values of $B^{\ell}_{rha}$ are known. Cache and backhaul capacity constraints are imposed by \eqref{MP_SC} and \eqref{MP_BC}. Constraint \eqref{MP_one} indicate that exactly one of the columns with respect to content $i$ and cache server $h$ is selected.

\subsection{Column Generation} \label{Sec:AD:Subsec:CG}

There are exponentially many candidate columns. However, the structure of \eqref{MP} can be exploited by using column generation. In column generation, the problem is decomposed into a restricted master problem (RMP) and a subproblem (SP), and RMP and SP are solved alternately. A new column that has the potential of improving the objective function is generated in each iteration by solving the SP. Note that most of the variables $\chi_{hi\ell}$ are zeros at optimum due to \eqref{MP_one}, and consequently only a small number of columns will be generated.

\subsubsection{Restricted Master Problem} \label{Sec:AD:Subsec:CG:Subsubsec:RMP}

RMP is a linear programming (LP) problem that differs from \eqref{MP} in the following: 
\begin{enumerate}
	\item The binary variables $\boldsymbol{y}$ and $\boldsymbol{\chi}$ are relaxed using the continuous domain $[0,1]$.
	\item For each content $i$ and server $h$, only a subset of the $\chi$-variables is present. This is represented by defining $\mathcal{L}^\prime_{hi}$, containing the columns generated so far for $i$ and $h$. Initially, $\mathcal{L}^\prime_{hi}$ only contains the columns representing that content $i$ is not cached by cache server $h$.
	\item As the subset of columns in RMP are differed by $i$ and $h$, we use $B^{i\ell}_{rha}$, $q_{hit\ell}$, and $p_{hit\ell}$ to replace $B^{\ell}_{rha}$, $q_{t\ell}$, and $p_{t\ell}$ for any column $\ell$ in individual subtset $\mathcal{L}^\prime_{hi}$.
\end{enumerate}

\subsubsection{Subproblem} \label{Sec:AD:Subsec:CG:Subsubsec:SP}

At the optimum of RMP, denote by $\pi^*_{rha}$, $\mu^*_{ht}$, $\varphi^*_{ht}$, and $\lambda^*_{hi}$ the corresponding optimal dual variables of the counterparts of \eqref{MP_y}, \eqref{MP_SC}, \eqref{MP_BC}, and \eqref{MP_one}, respectively. For any content $i$ and cache server $h$, the LP reduced cost of column $\ell$ is 
\begin{align}
	S_{hi\ell} + \sum_{r \in \mathcal{R}} \sum_{a = 0}^{t_r-1} B^{i\ell}_{rha} \pi^*_{rha} - s_i \sum_{t \in \mathcal{T}} \left( q_{hit\ell}\mu^*_{ht} + p_{hit\ell}\varphi^*_{ht} \right) - \lambda^*_{hi}.
\end{align}

By LP theory, only a column with negative reduced cost may improve the objective function value. To this end, we would like to find the column outside $\mathcal{L}^\prime_{hi}$ of minimum reduced cost. Obviously, this SP is decomposed by cache server as well as content. Finding the minimum reduced cost can be equivalently expressed using variables $\boldsymbol{x}$ and $\boldsymbol{z}$ in the SP, formulated below for content $i$ and server $h$.
\begin{subequations}\label{SP}
	\begin{align}
		\text{[SP$_{hi}$]\ } &\underset{\boldsymbol{x},\boldsymbol{z} \in \{0,1\}}\min \ \ S_{hi}(\boldsymbol{x}, \boldsymbol{z})+ \sum_{r\in R_{hi}} \sum_{a=0}^{t_r - 1} \pi^*_{rha} x_{hit_ra}\notag\\ &\qquad\qquad - s_i \sum_{t\in \mathcal{T}}(\mu^*_{ht} z_{hit} + \varphi^*_{ht} x_{hit0})- \lambda^*_{hi}\label{SP_obj}\\
		\text{s.t.} \ \ \ & 	\text{\eqref{ILP_c1}-\eqref{ILP_c3}},\notag
	\end{align}
\end{subequations}
where we use $S_{hi}(\boldsymbol{x}, \boldsymbol{z})$ to represent that the $S$-value of content $i$ for cache server $h$ is a function of $\boldsymbol{x}$ and $\boldsymbol{z}$. 

\subsection{SP as a Shortest Path Problem} \label{Sec:AD:Subsec:Fig}

\input{fig-shortest_path_graph}

We show that SP$_{hi}$ can be solved as a shortest path problem in polynomial time \cite{10.5555/137406}, even though it is stated as an ILP. Consider content $i$ and cache server $h$. We construct an acyclic directed graph in Fig. \ref{fig:shortest_path_graph} in which finding the shortest path from node $\Lambda$ to node $\Pi$ is equivalent to solving SP$_{hi}$. For brevity, we define an auxiliary function 
\begin{equation}\label{wight1}
	g_{hi}(o, d, a) = \sum_{ r \in \mathcal{R}_{\varsigma}}\pi^*_{rha},
\end{equation}
where $\mathcal{R}_{\varsigma} = \{r \in \mathcal{R}^{\text{m}}| (h \in \mathcal{H}_r) \wedge (i_r = i) \wedge (o_r = o) \wedge (d_r \geq d) \}$. A request in set $\mathcal{R}_{\varsigma}$ has the following properties:
\begin{enumerate}
	\item The request is made for content $i$ in time slot $o$.
	\item The request is MCR, and it can be satisfied by cache server $h$.
	\item The deadline of the request is greater than or equal to $d$.
\end{enumerate}

In the figure, the nodes can be divided into upper and lower parts except for $\Lambda$ and $\Pi$:
\begin{itemize}
	\item The upper nodes ($\theta$-nodes) represent that the content is not cached. One can see that we construct multiple nodes instead of only one node to represent the uncached status in a time slot. The subscripts are used to distinguish between the nodes in the same time slot. There are $\frac{t^2 - t + 2}{2}$ $\theta$-nodes in time slot $t$, and for these $\theta$-nodes, the predecessors of the last $(t-1)$ $\theta$-nodes are $\kappa$-nodes, for example, in time slot $3$, $\kappa^{2}_{0}\rightarrow \theta^{3}_{3}$ and $\kappa^{2}_{1}\rightarrow \theta^{3}_{4}$. The predecessors of other $\theta$-nodes are $\theta$-nodes in the previous time slot, for example, $\theta^{2}_{1}\rightarrow \theta^{3}_{1}$ and $\theta^{2}_{2}\rightarrow \theta^{3}_{2}$.
	\item The lower nodes ($\kappa$-nodes) represent that the content is cached. Similarly, the superscripts of $\kappa$-nodes indicate their associated time slots. The subscripts are the AoI values of the content (therefore, the subscripts start with zero). 
\end{itemize}
We have the following types of arcs:
\begin{enumerate}
	\item The gray dashed arcs all terminate at $\theta$-nodes, and they represent that the content is not cached in the next time slot. For example, arcs $(\theta_{2}^{2}, \theta^{3}_{2})$ and $(\kappa_{0}^{2}, \theta^{3}_{3})$ both mean that the content is not cached in time slot $3$. The weights of these arcs are zero as the corresponding the $x$-variables and the $z$-variables are both zero in \eqref{SP_obj}.
	\item The orange arcs all terminate at node $\Pi$ and emanate from the nodes that are in the last time slot. These arcs are relevant to calculating the corresponding $S_{hi}(\boldsymbol{x}, \boldsymbol{z})$ and constant term $\lambda_{hi}^*$ in \eqref{SP_obj}. The weight of any of these arcs is 
	\begin{equation}\label{orangearc}
		\left|  \mathcal{R}^{\text{s}}_{hi} \right|\left[ f_{i}\left(0\right) + \alpha s_i \right]-\lambda_{hi}.
	\end{equation} 
	The first term of \eqref{orangearc} is for all the SCRs in $\mathcal{R}^{\text{s}}_{hi} $. If some of these requests are met by the cache, the costs in \eqref{orangearc} will be subtracted, see further details below.
	\item A blue arc connects two adjacent $\kappa$-nodes, both with AoI zero. Such an arc corresponds to that the content is updated in the two adjacent time slots. In this case, the corresponding $x$-variables and $z$-variables are both one in \eqref{SP_obj}. For arc $(\kappa_{0}^{t-1}, \kappa_{0}^{t})$, the weight is
	\begin{equation}\label{bluearc}
		-\left|\mathcal{R}_{\Xi}\right| \alpha s_i + g_{hi}(t, t, 0) - s_i (\mu^*_{ht} + \varphi^*_{ht}),
	\end{equation}
	where $\mathcal{R}_{\Xi} = \{r \in \mathcal{R}^{\text{s}}_{hi}| d_r = t\}$. A request in $\mathcal{R}_{\Xi}$ is made for content $i$, and it can be satisfied by cache server $h$ before time slot $t$. Clearly, if we select blue arc $(\kappa_{0}^{t-1}, \kappa_{0}^{t})$, the SCRs in $\mathcal{R}_{\Xi}$ can be satisfied in time slot $t$ by cache server instead of the cloud server with AoI zero. Therefore, the first term of \eqref{bluearc} has the effect of subtracting the download costs present in \eqref{orangearc}. In addition, term $g_{hi}(t,t,0)$ is because the MCRs in $\{r \in \mathcal{R}^{\text{m}} | (h \in \mathcal{H}_r) \wedge (i_r = i) \wedge (o_r = t) \wedge (d_r \geq t)  \}$ can be satisfied with AoI zero. 
	\item A black arc emanates from a $\kappa$-node with nonzero AoI value, and terminates at a $\kappa$-node with zero AoI value. If such a black arc $(\kappa_{a}^{t-1}, \kappa_{0}^{t})$ is selected, it means that the content has not been updated for $a$ time slots, and it is updated in time slot $t$. In this case, the MCRs in $\{r \in \mathcal{R}^{\text{m}}| (h \in \mathcal{H}_r) \wedge (i_r = i) \wedge (t-a \leq o_r \leq t) \wedge (d_r \geq t) \}$ can be satisfied in time slot $t$ with AoI zero, so the weight of arc $(\kappa_{a}^{t-1}, \kappa_{0}^{t})$ is given by
	\begin{equation}\label{blackarc}
		-\left|\mathcal{R}_{\Xi}\right| \alpha s_i + \sum_{\tau = t - a}^{t} g_{hi}(\tau, t, 0) - s_i (\mu^*_{ht} + \varphi^*_{ht}).
	\end{equation}
	\item A red arc emanates from a $\theta$-node, and terminates at a $\kappa$-node with zero AoI. If such a red arc $(\theta_{n}^{t-1}, \kappa_{0}^{t})$ is selected, it means that the content is updated in time slot $t$, and also it has not been updated for some consecutive time slots. We give the following recursive function denoted by $\xi(\theta_{n}^{t-1})$ that calculates how many consecutive time slots the content has not been updated at node $\theta_{n}^{t-1}$:
	\begin{equation}
		\xi(\theta^{t-1}_{n})=\left\{
		\begin{array}{lcl}
		1   &    &  { pre(\theta^{t-1}_{n}) = \Lambda}\\
		\xi(\theta^{t-2}_{n}) + 1   &    &  { pre(\theta^{t-1}_{n}) = \theta^{t-2}_{n}}\\
		a+1  &   &  {pre(\theta^{t-1}_{n}) = \kappa_a^{t-2}}\\
		\end{array} \right.
	\end{equation}
	where $pre(\cdot)$ is the predecessor of some node. For example, $\xi(\theta^{4}_{4}) = \xi(\theta^{3}_{4}) + 1 = (1+1)+1 = 3$ represents the content has not been updated for three time slots (i.e., time slots $2$, $3$, and $4$).
	With the auxiliary function, the weight of arc $(\theta_{n}^{t-1}, \kappa_{0}^{t})$, similar to \eqref{blackarc}, is given by 
	\begin{equation}
		-\left|\mathcal{R}_{\Xi}\right| \alpha s_i +  \sum_{\tau = t - \xi(\theta^{t-1}_{n})}^{t} g_{hi}(\tau, t, 0) - s_i (\mu^*_{ht} + \varphi^*_{ht}).
	\end{equation}
	\item The purple arcs represent AoI evolution over time slots when the content is cached but not updated. For these arcs, the corresponding $x$-variables are zero, but the corresponding $z$-variables are one. The weight of arc $(\kappa_{a-1}^{t-1}, \kappa_{a}^{t})$ is given by
	\begin{align}\label{purplearcs}
		 &\sum_{\tau = 0}^{t-1} \sum_{r \in \mathcal{R}_{\Psi}} \left[ -f_i\left(0\right)  
		 - \alpha s_i + f_i\left(\max\left\{0, a - \tau\right\} \right) \right] \notag\\
		 &\qquad\qquad\qquad\qquad\qquad\quad + g_{hi}(t, t, a) - s_i \mu^*_{ht},
	\end{align}
	where $\mathcal{R}_\Psi  = \{r \in \mathcal{R}^{\text{s}}_{hi}| (o_r = t - \tau) \wedge (d_r = t)\}$, i.e., the requests in $\mathcal{R}_\Psi $ are SCRs made for content $i$ in time slot $(t-\tau)$, and they can be satisfied by cache server $h$ before time slot $t$. If arc $(\kappa_{a-1}^{t-1}, \kappa_{a}^{t}  )$ is selected, the AoI values in the previous $a$ time slots must be $0, 1,\dots, a-1$. Thus we know that the lowest possible AoI value for the SCRs in $\{r \in \mathcal{R}^{\text{s}}_{hi}| (o_r \leq t- a) \wedge (d_r = t)\}$ must be zero, and for the SCRs in $\{r \in \mathcal{R}^{\text{s}}_{hi}| (o_r > t- a) \wedge (d_r = t)  \}$,  the lowest possible AoI value must be $a - (d_r-o_r)$, therefore the first term of \eqref{purplearcs}. In addition, the MCRs in $\{ r \in \mathcal{R}^{\text{m}}| (h \in \mathcal{H}_r) \wedge (i_r = i) \wedge (o_r = t) \wedge (d_r \geq t) \}$ can be satisfied with AoI $a$, hence $g_{hi}(t, t, a)$.
\end{enumerate}

\begin{theorem}
	For content $i$ and cache server $h$, SP{$_{hi}$} can be solved in polynomial time as a shortest path problem.
\end{theorem}
\begin{proof}
	We show that the optimal solution of SP{$_{hi}$} can be obtained as the shortest path in the graph constructed. Suppose the optimal solution of SP$_{hi}$, i.e., $\boldsymbol{x}^*$ and $\boldsymbol{z}^*$, is given. A corresponding path is constructed as follows. 
	
	\begin{enumerate}
		\item Starting with time slot one:\
		\begin{itemize}
			\item If $z_{hi1}^* = 0$ and $x_{hi10}^* = 0$, arc $(\Lambda, \theta^{1}_{1})$ is selected.
			\item If $z_{hi1}^* = 1$ and $x_{hi10}^* = 1$, arc $(\Lambda, \kappa_{0}^{1})$ is selected.
		\end{itemize}
		\item For time slot $t$ $(t>1)$:\
		\begin{itemize}
			\item If $z_{hit}^* = 0$ and $x_{hita}^* = 0$ $ \forall a \in \left\{0, 1, \dots ,t-1\right\}$, arc $(\theta^{t-1}_{n}, \theta^{t}_{n})$ or $(\kappa^{t-1}_{a}, \theta^{t}_{n})$ is selected according to which one applies in the previous time slot.
			\item If $z_{hit}^* = 1$ and $x_{hit0}^* = 1$, arc $(\theta^{t-1}_{n}, \kappa_{0}^{t})$ or $(\kappa_{a}^{t-1}, \kappa_{0}^{t})$ is selected according to the previous time slot. 
			\item If $z_{hit}^* = 1$ and $x_{hita}^* = 1$ $(a \neq 0)$, arc $(\kappa^{t-1}_{a-1}, \kappa_{a}^{t})$ is selected. 
		\end{itemize}
		\item Finally, one of the orange arcs $(\theta^{t}_{n}, \Pi)$ and $(\kappa^{t}_{a}, \Pi)$ is selected to reach node $\Pi$.
	\end{enumerate}
	Clearly, a path from node $\Lambda$ to $\Pi$ is constructed in this way. According to \eqref{orangearc}-\eqref{purplearcs}, the cost of the derived path equals the object function value of  SP$_{hi}$ for $\boldsymbol{x}^*$ and $\boldsymbol{z}^*$. 
	
	Conversely, consider the shortest path. For time slot $t$:
	\begin{itemize}
		\item If the path contains any $\theta$-node in time slot $t$, we set $z_{hit} = 0$ and $x_{hita} = 0$ $(\forall a\in \{0,1, \dots, t-1\})$. 
		\item If the path contains node $\kappa^{t}_{\bar{a}}$ in time slot $t$, we set $z_{hit} = 1$, $x_{hit\bar{a}} = 1$, and $x_{hita} = 0$ $(a \neq \bar{a})$. 
	\end{itemize}
	With this variable value setting, the objective function has the same value as the path length (i.e., total arc weights). One can observe that the solution constructed in this way satisfies the constraints of the SP because any path from node $\Lambda$ to $\Pi$ adheres to constraints \eqref{ILP_c1} and \eqref{ILP_c3}. 	
	
	In addition, the shortest path problem in directed acyclic graph can be solved in polynomial time no matter if the arc costs are negative or not \cite{10.5555/137406}. There are $\mathcal{O}(T^3) $ nodes and arcs in the corresponding graph, therefore, the complexity of solving the shortest path problem is $\mathcal{O}(T^3)$. Hence the conclusion.
\end{proof}
The column generation algorithm (CGA) is summarized in Algorithm \ref{al:cg}. The algorithm first initializes $\mathcal{L}'_{hi}$ $\forall h\in \mathcal{H}, i \in \mathcal{I}$ with all-zero column. Then CGA solves RMP and obtains the optimal values of its dual variables. For each content $i$ for cache server $h$, CGA will find the column with minimal reduced cost via solving the corresponding shortest path problem, and add it to $\mathcal{L}'_{hi}$ if the reduced cost is negative. The algorithm will terminate if all reduced costs are non-negative. 
	
\begin{algorithm}[tbp]
    \DontPrintSemicolon
    \caption{Column generation algorithm} \label{al:cg}
    \KwIn{$\mathcal{H}$, $\mathcal{I}$, $\mathcal{T}$, $C_h$, $G_h$, $\forall h \in \mathcal{H}$, $s_i$, $f_i(a)$, $\forall i \in \mathcal{I}$} 
    \KwOut{$\boldsymbol{\chi}^*$, $\boldsymbol{y}^*$, $ \mathcal{L}^\prime$}   
    Initialize $ \mathcal{L}^\prime$\\
    \Repeat
    {
    	\rm The optimum of \eqref{SP_obj} $\geq 0$, $\forall h \in \mathcal{H}, i \in \mathcal{I}$
    }
    {
    	Solve RMP to obtain the optimum $(\boldsymbol{\chi}^*, \boldsymbol{y}^*)$ and dual optimum $(\boldsymbol{\pi}^*, \boldsymbol{\mu}^*, \boldsymbol{\varphi}^*, \boldsymbol{\lambda}^* )$\\
    		\For{$h \in \mathcal{H}$, $i \in \mathcal{I}$ \label{cgalpara1}}
    		{
        		Solve SP$_{hi}$ \eqref{SP} by a shortest path algorithm\\
    	    	\If{\rm the optimum of \eqref{SP_obj} $< 0$} 
        		{
        			Add the corresponding column to $\mathcal{L}^\prime$\\
    	    	}
        	}
    }   
    \Return {$\boldsymbol{\chi}^*$, $\boldsymbol{y}^*$, $ \mathcal{L}^\prime$}\\
\end{algorithm} 

\subsection{Rounding Algorithm} \label{Sec:AD:Subsec:RS}

We remark that Algorithm \ref{al:cg} ends with a solution (i.e., $\boldsymbol{\chi}^*$) that is not necessarily integer. At this stage, the objective function value is a lower bound (LB) of the global optimum. To obtain an integer solution, a naive rounding strategy is to fixe some fractional $\chi$-variable to be one or zero. However, this choice is quite aggressive because the decisions over all time slots for the corresponding content become all fixed when one $\chi$-variable is fixed, and thus we have no chance to alter the solution of this content at all.

We first define two following indicators that can be viewed as the likelihood of caching and updating content $i$ for cache server $h$ in time slot $t$, respectively. 
\begin{align}
	\Gamma_{hit} &= \sum _{\ell\in \mathcal{L}^\prime}q_{hit\ell} \chi_{hi\ell}^*, \label{Gamma} \\
	\Omega_{hit} &= \sum _{\ell\in \mathcal{L}^\prime}p_{hit\ell} \chi_{hi\ell}^*. \label{Omega}
\end{align}
In the light of the use of disjunction in branch-and-bound in solving integer programs \cite{karamanov2011branching}, we consider rounding the $\Gamma$-variables and the $\Omega$-variables instead of the $\chi$-variables. In the following, we prove a relationship among $\boldsymbol{\chi}^*$, $\boldsymbol{\Gamma}$, and $\boldsymbol{\Omega}$, and then propose a rounding algorithm based on this result.

\begin{theorem}\label{theorem 3}
	$\boldsymbol{\chi}^*$ is integer if and only if $\boldsymbol{\Gamma}$ and $\boldsymbol{\Omega}$ are integers.
\end{theorem}
\begin{proof}
	The necessity is obvious by the definition of $\boldsymbol{\Gamma}$ and $\boldsymbol{\Omega}$. For sufficiency, assume $\boldsymbol{\Gamma}$ and $\boldsymbol{\Omega}$ are integer. Suppose $\boldsymbol{\chi}^*$ is fractional, i.e., for some $\chi^*_{hi\ell}$, we have $0<\chi^*_{hi\ell} < 1$. By the counterpart of \eqref{MP_one} in RMP, we know that at least one more fractional $\chi$-variable for content $i$ and cache server $h$ must exist. Suppose we have two columns $\ell_1, \ell_2 \in \mathcal{L}'_{hi}$, and their corresponding $\chi_{hi\ell_1}^*$ and $\chi_{hi\ell_2}^*$ are such that 
	\begin{itemize}
		\item $\chi_{hi\ell_1}^*$ and $\chi_{hi\ell_2}^*$ are both fractional;
		\item $\chi_{hi\ell_1}^* + \chi_{hi\ell_2}^* = 1$.
	\end{itemize}
	Then we can infer that 
	\begin{equation}
		\Gamma_{hit} = \sum _{\ell\in \mathcal{L}^\prime}q_{hit\ell} \chi_{hi\ell}^*  = q_{hit\ell_1}\chi_{hi\ell_1}^* +  q_{hit\ell_2}\chi_{hi\ell_2}^*.
	\end{equation}
    Because every column is unique in $\mathcal{L}^\prime_{hi}$, we can always find some time slot $t$ in which $q_{hit\ell_1} = 1$, $q_{hit\ell_2} = 0$ or $q_{hit\ell_1} = 0$, $q_{hit\ell_2} = 1$. In this time slot, $\Gamma_{hit}$ equals $q_{hit\ell_1}\chi_{hi\ell_1}^*$ or $q_{hit\ell_2}\chi_{hi\ell_2}^*$, respectively, and thus $\Gamma_{hit}$ must be fractional, which contradicts the assumption. The argument is easily generalized to more than two columns. Moreover, the same argument applies to $\boldsymbol{\Omega}$. Hence the conclusion.
\end{proof}

By Theorem \ref{theorem 3}, we can perform rounding on $\boldsymbol{\Gamma}$ and $\boldsymbol{\Omega}$ after column generation, instead of $\boldsymbol{\chi}^*$. That is, we fix some $\Gamma$-variables and $\Omega$-variables to be zero or one, then the columns violating these fixed values will be removed from $\mathcal{L}^\prime_{hi}$. This is much less aggressive in rounding than fixing $\chi$-variables because the decision of only one time slot is made, instead of over all time slots. The rounding algorithm is summarized in Algorithm \ref{al:RA}. More specifically, Algorithm \ref{al:RA} has four stages:
\begin{enumerate}
	\item Steps \ref{Freeze1}-\ref{Freeze2}: We first fix the entities that are integer. Namely, if $\Omega_{hit} = 1$, we fix $\Omega_{hit}$ and $\Gamma_{hit} $ to be both one. Similarly, we fix $\Omega_{hit}$ and $\Gamma_{hit} $ to be both zero if $\Gamma_{hit} = 0$. 
	\item Steps \ref{Round1}-\ref{Round2}: For rounding, we first consider fractional $\boldsymbol{\Omega} $ as it impacts both cache and backhaul capacities. We find the $\Omega_{h\hat{i}\hat{t}}$ $(\hat{i}\in\mathcal{I}, \hat{t}\in\mathcal{T})$ that is closest to zero or one (Step \ref{closest}) for each cache server $h$ $(h\in\mathcal{H})$. In the former case, we fix $\Omega_{h\hat{i}\hat{t}}$ to be zero. In the latter case, if the remaining capacities admit content size $s_{\hat{i}}$, we fix both $\Omega_{h\hat{i}\hat{t}}$ and $\Gamma_{h\hat{i}\hat{t}}$ to be one. 
	\item Steps \ref{Round3}-\ref{Round5}: If $\boldsymbol{\Omega} $ is integer, we subsequently consider fractional $\boldsymbol{\Gamma} $, if any. Similarly, we first fix the integer entities (Steps \ref{Freeze3}-\ref{Freeze4}), i.e., fix $\Gamma_{hit}$ to be one if $\Gamma_{hit} = 1$. Then we find the $\Gamma_{h\hat{i}\hat{t}}$ $(\hat{i}\in\mathcal{I}, \hat{t}\in\mathcal{T})$ that is closest to being an integer for each cache server $h$ $(h\in\mathcal{H})$, and perform rounding as above. 
	\item Steps \ref{Remove1}-\ref{Remove2}: Finally, we recalculate the remaining backhaul and cache capacities, and remove all the incompatible columns from $\mathcal{L}^\prime_{hi}$. The incompatible columns include the two following types:
	\begin{itemize}
		\item The columns violating the rounding decisions, e.g., we need to remove all $\ell$ in which $ p_{123\ell} = 0$ if we have fixed $\Omega_{123}$ to be one.
		\item The columns corresponding to updating or caching a content in some time slot but this would exceed the backhaul or cache capacity.
	\end{itemize}
\end{enumerate}

\begin{algorithm}[tbp]
    \caption{Rouding algorithm} \label{al:RA}
    \KwIn{$\boldsymbol{\chi}^*$, $C_h$, $G_h$, $s_i$, $ \mathcal{L}^\prime_{hi}$, $\forall h \in \mathcal{H}, i \in \mathcal{I}$} 
    \KwOut{$ \mathcal{L}^\prime_{hi}$, $\forall h \in \mathcal{H}, i \in \mathcal{I}$}   
    Compute $\boldsymbol{\Gamma}$ and $\boldsymbol{\Omega}$ by \eqref{Gamma} and \eqref{Omega}\label{Freeze1}\\ 
    \For{$h \in \mathcal{H}$, $i \in \mathcal{I}$, $t \in \mathcal{T}$ \label{RApara1}}
    {
    	\If( \tcp*[h]{$\Gamma_{hit}$ must be one}){$\Omega_{hit} = 1$}
    	{
    		Fix both $\Omega_{hit}$ and $\Gamma_{hit}$ to be one\\
    		Remove all the nodes in time slot $t$ except node $\kappa_{0}^{t}$ and their arcs from the graph of SP$_{hi}$\label{RmvNdofGrph1}
    	}
    	\If( \tcp*[h]{$\Omega_{hit}$ must be zero}){$\Gamma_{hit} = 0$}
    	{
    		Fix both $\Omega_{hit}$ and $\Gamma_{hit}$ to be zero\\
    		Remove all the $\kappa$-nodes in time slot $t$ and their arcs from the graph of SP$_{hi}$\label{RmvNdofGrph2}
    	}
    }
    \For{$h \in \mathcal{H}$, $t \in \mathcal{T}$ \label{RApara2}}
    {
    	Compute the remaining backhaul and cache capacities, denoted by $\bar{G}_{ht}$ and $\bar{C}_{ht}$, respectively\label{Freeze2}\\
    }
    
    \For{$h \in \mathcal{H}$\label{Round1}} 
    {
    	\uIf{$\exists \Omega_{hit} \notin \mathbb{Z}$ $(i \in \mathcal{I},t \in \mathcal{T})$ }
    	{
    		$(h, \hat{i}, \hat{t}) \leftarrow \arg\min \{\Delta_{hit}| (\Delta_{hit} \!=\! \min\{\Omega_{hit}, 1 - \Omega_{hit} \}) \wedge (\Omega_{hit} \notin \mathbb{Z} ) \wedge ( i \in \mathcal{I}, t \in \mathcal{T} )  \}$\label{closest}\\
    		\uIf{$(\Omega_{h \hat{i} \hat{t}} < 0.5) \vee (s_{\hat{i}} > \bar{G}_{h\hat{t}})$}
    		{
    			Fix $\Omega_{h \hat{i} \hat{t}}$ to be zero\\
    			Remove node $\kappa_{0}^{\hat{t}}$ and its arcs from the graph of SP$_{h\hat{i}}$\label{RmvNdofGrph3}
    		}
    		\Else
    		{
    			Fix both $\Omega_{h\hat{i}\hat{t}}$ and $\Gamma_{h\hat{i}\hat{t}}$ to be one\\
    			Remove all the nodes in time slot $\hat{t}$ except node $\kappa_{0}^{\hat{t}}$ and their arcs from the graph of SP$_{h\hat{i}}$ \label{Round2}
    		}
    	}
    	\ElseIf{$\exists \Gamma_{hit} \notin \mathbb{Z}$ $(i \in \mathcal{I}, t \in \mathcal{T})$\label{Round3}}
    	{
    		\For{$i \in \mathcal{I}$, $t \in \mathcal{T}$ \label{Freeze3}}
    		{
    			\If{$\Gamma_{hit} = 1$}
    			{
    				Fix $\Gamma_{hit}$ to be one\\
    				Remove all the $\theta$-nodes in time slot $t$ and their arcs from the graph of SP$_{hi}$   \label{Freeze4}
    			}
    		}
    		\For{$t \in \mathcal{T}$}
    		{
    			Compute $C_{ht}'$\\
    		}
    	
    		$(h, \hat{i}, \hat{t}) \leftarrow \arg\min \{\Delta_{hit}| (\Delta_{hit} \!=\! \min\{\Gamma_{hit}, 1 - \Gamma_{hit} \}) \wedge (\Gamma_{hit} \notin \mathbb{Z} ) \wedge (i \in \mathcal{I}, t \in \mathcal{T} )  \}$\\
    		\uIf{$(\Gamma_{h \hat{i} \hat{t}} < 0.5) \vee (s_{\hat{i}} > \bar{C}_{h\hat{t}})$ \label{Round4}}
    		{
    			Fix $\Gamma_{h \hat{i} \hat{t}}$ to be zero\\
    			Remove all the $\kappa$-nodes in time slot $\hat{t}$ and their arcs from the graph of SP$_{h\hat{i}}$ \label{RmvNdofGrph6}
    		}
    		\Else
    		{
    			Fix $\Gamma_{h \hat{i} \hat{t}}$ to be one\\
    			Remove all the $\theta$-nodes in time slot $\hat{t}$ and their arcs from the graph of SP$_{h\hat{i}}$ \label{Round5} 
    		}
    	}
    }
    
    \For{$h \in \mathcal{H}$, $t \in \mathcal{T}$ \label{Remove1}} 
    {
    	Compute $\bar{G}_{ht}$ and $\bar{C}_{ht}$\\
	    Remove all incompatible columns from $ \mathcal{L}^\prime_{hi} \forall i \in \mathcal{I}$ and the corresponding $\chi$-variables\label{Remove2}\\
    }
    
    \Return {$ \mathcal{L}^\prime_{hi}$, $\forall h \in \mathcal{H}, i \in \mathcal{I}$}\\
\end{algorithm} 

\subsection{Repeated Column Generation Algorithm} \label{Sec:AD:Subsec:RCGA}

After a rounding operation, there may exist additional columns that will improve the RMP. Therefore we perform our column generation algorithm and rounding algorithm alternately to obtain an integer solution. After the alternative process terminates, all the $\Omega$ and $\Gamma$-variables are integers, and we can obtain an integer solution $\boldsymbol{\chi}^*$ by Theorem \ref{theorem 3}. Once integer $\boldsymbol{\chi}^*$ is found, the corresponding optimal values of binary $y$-variables are determined in constant time $\mathcal{O}(|\mathcal{R}^{\rm m}|)$.

We name the overall algorithm repeated column generation algorithm (RCGA), summarized in Algorithm \ref{al:RCGA}. Note that Steps \ref{RmvNdofGrph1}, \ref{RmvNdofGrph2}, \ref{RmvNdofGrph3}, \ref{Round2}, \ref{Freeze4}, \ref{RmvNdofGrph6}, and \ref{Round5} of Algorithm \ref{al:RA} are in fact for RCGA. These steps remove some nodes and their arcs in the graph of SPs such that the columns generated by subsequent SPs comply to the rounding decisions. As there are $H$ cache servers, $I$ contents, and $T$ time slots, and at least $H$ elements of the $\Omega$-variables or the $\Gamma$-variables become integer in each iteration, the process is executed in at most $I \times T$ iterations.

\begin{algorithm}[tbp]
    \caption{The framework of RCGA} \label{al:RCGA}
    \KwIn{$\mathcal{H}$, $\mathcal{I}$, $\mathcal{T}$, $C_h$, $G_h$, $\forall h \in \mathcal{H}$, $s_i$, $f_i(a)$, $\forall i \in \mathcal{I}$} 
    \KwOut{$\boldsymbol{\chi}^*$, $\boldsymbol{y}^*$, $ \mathcal{L}^\prime_{hi}$, $\forall h \in \mathcal{H}, i \in \mathcal{I}$}   
    Apply Algorithm \ref{al:cg} to obtain $\boldsymbol{\chi}^*$\\
    \While{$\boldsymbol{\chi}^* \notin \boldsymbol{\mathbb{Z}}^{|\mathcal{L}^\prime_{hi}|}$}
    {
    	Apply Algorithm \ref{al:RA}\label{RARA}\tcp{rounding operation}
    	Apply Algorithm \ref{al:cg} to obtain $\boldsymbol{\chi}^*$
    }
    Obtain $\boldsymbol{y}^*$ when $\boldsymbol{\chi}^*$ is given\tcp{easy to solve}
    \Return {$\boldsymbol{\chi}^*$, $\boldsymbol{y}^* $, $ \mathcal{L}^\prime$}\\
\end{algorithm} 


\section{Performance Evaluation} \label{Sec:evaluation}

\subsection{Preliminaries}
In this section, we show performance results of RCGA. We consider MCSP instances of both three and sevel cells. The requests for contents follow a binomial distribution $B(I, 0.5)$. Let $\rho^m \in (0,1)$ be the percentage of MCRs. In our simulation, the MCRs have either two or three candidate cache servers. Denote the ratio between the number of MCRs with three candidate servers to that of two candidate servers by $\rho^{\rm tt}$. We will present the numerical results for various values of $\rho^m$ and $\rho^{\rm tt}$. The cache capacity is set to be $C_h = \frac{1}{2}\sum_{i\in \mathcal{I}} s_i$ $(\forall h \in \mathcal{H})$. We are more interested in the impact of backhaul capacity as it is a bottleneck for AoI, so we set it to be $G_h = G =  \rho^b \sum_{i\in \mathcal{I}} s_i$ $(\forall h \in \mathcal{H})$, and vary $\rho^b$ in interval $[0.1, 0.5]$. We take $f_i(a) = \mathrm{e}^a $ $(a \in  \mathbb{N})$ as our AoI cost function, and the function is in fact the reciprocal of an AoI utility function in \cite{5062058}. In addition, we set $\alpha = 11$ and $\beta = 1$. The size of contents is uniformly generated within $[1,10]$. 

\subsection{Benchmarks}
In our simulation, we use three benchmark schemes.

\subsubsection{PBA}
{
    We take the cost obtained by a heuristic algorithm named Popularity-based Algorithm (PBA) \cite{6708492} as a baseline scheme. In PBA, a content with many requests has higher priority of being cached and updated, and the cached contents will be subject to update when its AoI cost has reached half of its download cost. If the cache or backhaul capacity does not allow a content to be cached or updated, PBA will consider the next content by the prioritization.
}
\subsubsection{OPT}
{
    This baseline scheme is the optimum of ILP \eqref{formulation} obtained by Gurobi optimizer \cite{GUROBI}. With this benchmark, we can accurately measure the deviation of RCGA and PBA from the optimum. However, it is computationally difficult to obtain global optimum for the 7-cell instances, so we use it for evaluation only for 3-cell instances. 
}
\subsubsection{LB}
{
    We also take the LB obtained by performing column generation without rounding as a reference value. It is a valid comparison because the deviation with respect to the global optimum will never exceed the deviation from the LB. We will see that, numerically, using the LB remains accurate in gauging the gap from the optimum.
}

\pgfplotsset{width=1\textwidth, height=4.4 cm, compat=1.9}

\begin{figure*}
	\centering
    \begin{minipage}{0.48\textwidth}
        \begin{center}
    		\begin{tikzpicture}
        		\begin{axis}[
        			scaled y ticks=base 10:-3,
        		    ylabel={Cost (3-cell)},
        		    xmin=25, xmax=150,
        		    ymin=0, ymax=3500,
        		    xtick={25, 50, 75, 100, 125, 150},
        		    legend pos=north west,
        		    grid style=densely dashed,
        		    tick label style={font=\scriptsize},
        		    label style={font=\small},
        		    legend style={font=\scriptsize},
        		]
        		
            		\addplot[ color=red, mark=square, line width=0.8pt]     
            		coordinates { (25, 750) (50, 1050) (75, 1180) (100, 1250) (125, 1430) (150, 1480) };
            		\addplot[ color=blue, mark=o, mark options={solid}, line width=0.8pt]     
            		coordinates { (25, 850) (50, 1430) (75, 1550) (100, 1950) (125, 2380) (150, 3100) };
            		
            		\addplot[ color=black, mark=triangle, densely dashed, mark options={solid}, line width=0.8pt]
            		coordinates { (25, 730) (50, 1030) (75, 1160) (100, 1200) (125, 1380) (150, 1400) };
            		
            		\legend{RCGA, PBA, OPT}
        		
        		\end{axis}
    		\end{tikzpicture}
    		
    		\begin{tikzpicture}
        		\begin{axis}[
        			scaled y ticks=base 10:-3,
        		    xlabel={The number of contents $I$},
        		    ylabel={Cost (7-cell)},
        		    xmin=50, xmax=300,
        		    ymin=0, ymax=18000,
        		    xtick={50, 100, 150, 200, 250, 300},
        		    legend pos=north west,
        		    grid style=densely dashed,
        		    tick label style={font=\scriptsize},
        		    label style={font=\small},
        		    legend style={font=\scriptsize},
        		]
        		
            		\addplot[ color=red, mark=square, line width=0.8pt]     
            		coordinates { (50, 2921) (100, 4253) (150, 4985) (200, 5229) (250, 6384) (300, 6711) };
            		\addplot[ color=blue, mark=o, mark options={solid}, line width=0.8pt]     
            		coordinates { (50, 3500) (100, 5500) (150, 7000) (200, 8200) (250, 11200) (300, 16500) };
            		
            		\addplot[ color=gray, mark=triangle, densely dashed, mark options={solid}, line width=0.8pt]
            		coordinates { (50, 2821) (100, 4133) (150, 4785) (200, 5029) (250, 6184) (300, 6511) };
            		
            		\legend{RCGA, PBA, LB}
        		
        		\end{axis}
    		\end{tikzpicture}
            \captionsetup{font=footnotesize}
            \captionof{figure}{The average cost as function of the number of contents $I$ when $T = 12$, $G = 0.3 \sum_{i\in \mathcal{I}} s_{i}$, $\rho^{m} = 40\%$, and $\rho^{\rm tt} = 1:1$. The numbers of requests $R$ in the 3-cell and 7-cell instances are $500$ and $2000$, respectively.} \label{fig:1}
        \end{center}
    \end{minipage}
    \hspace{.15in}
    \begin{minipage}{0.48\textwidth}
        \begin{center}
    		\begin{tikzpicture}
        		\begin{axis}[
        			scaled y ticks=base 10:-3,
        		    ylabel={Cost (3-cell)},
        		    xmin=200, xmax=700,
        		    ymin=600, ymax=2700,
        		    xtick={200, 300, 400, 500, 600, 700},
        		    ytick={0, 1000, 2000, 3000},
        		    legend pos=north west,
        		    grid style=densely dashed,
        		    tick label style={font=\scriptsize},
        		    label style={font=\small},
        		    legend style={font=\scriptsize},
        		]
        		
            		\addplot[ color=red, mark=square, line width=0.8pt]     
            		coordinates { (200, 730) (300, 920) (400, 1180) (500, 1250) (600, 1450) (700, 1580) };
            		
            		\addplot[ color=blue, mark=o, mark options={solid}, line width=0.8pt]     
            		coordinates { (200, 1430) (300, 1620) (400, 1780) (500, 1950) (600, 2350) (700, 2580) };
            		
            		\addplot[ color=black, mark=triangle, densely dashed, mark options={solid}, line width=0.8pt]
            		coordinates { (200, 690) (300, 880) (400, 1120) (500, 1200) (600, 1400) (700, 1530) };
            		
            		\legend{RCGA, PBA, OPT}
        		
        		\end{axis}
    		\end{tikzpicture}
    		
    		\begin{tikzpicture}
        		\begin{axis}[
        			scaled y ticks=base 10:-3,
        		    xlabel={The number of requests $R$},
        		    ylabel={Cost (7-cell)},
        		    xmin=500, xmax=3000,
        		    ymin=1000, ymax=11000,
        		    xtick={500, 1000, 1500, 2000, 2500, 3000},
        		    legend pos=north west,
        		    grid style=densely dashed,
        		    tick label style={font=\scriptsize},
        		    label style={font=\small},
        		    legend style={font=\scriptsize},
        		]
        		
            		\addplot[ color=red, mark=square, line width=0.8pt]     
            		coordinates { (500, 2030) (1000, 3657) (1500, 4612) (2000, 5229) (2500, 5918) (3000, 6189) };
            		\addplot[ color=blue, mark=o, mark options={solid}, line width=0.8pt]     
            		coordinates { (500, 4600) (1000, 6100) (1500, 7500) (2000, 8200) (2500, 9000) (3000, 10000) };
            		
            		\addplot[ color=gray, mark=triangle, densely dashed, mark options={solid}, line width=0.8pt]
            		coordinates { (500, 1934) (1000, 3507) (1500, 4502) (2000, 5029) (2500, 5708) (3000, 5909) };
            		
            		\legend{RCGA, PBA, LB}
        		
        		\end{axis}
    		\end{tikzpicture}
        \captionsetup{font=footnotesize}
        \captionof{figure}{The average cost as function of the number of requests $R$ when $T = 12$, $G = 0.3 \sum_{i\in \mathcal{I}} s_{i}$, $\rho^{m} = 40\%$, and $\rho^{\rm tt} = 1:1$. The numbers of contents $I$ in the 3-cell and 7-cell instances are $100$ and $200$, respectively.} \label{fig:2}
        \end{center}
    \end{minipage}
\end{figure*}

\subsection{Rounding Strategy Performance}

In Section \ref{Sec:AD:Subsec:RS}, we claimed that a naive strategy that directly performs rounding on the $\chi$-variables might not be satisfactory. Table \ref{tab:rounding} shows a comparison of the solutions obtained by the naive rounding strategy (NRS) and RCGA for 3-cell and 7-cell scenarios with $100$ instances each. Success rate represents the percentage instances for which rounding leads to feasible solutions. It can be seen that NRS has a much lower success rate than RCGA. Column Gap in the table is the average gap of the feasible solutions from the optimum/LB. Obviously, RCGA is superior to NRS. 

\begin{table}[h]
    \caption{\label{tab:rounding}Comparison between NRS and RCGA}
    \begin{center}
        \begin{tabular}{*{4}{c}}
            \toprule
            \midrule
            {\bf Method} & {\bf Scale}& {\bf Success rate} & {\bf Gap} \\
            \midrule
            \multirow{2}{*}{NRS} & 3-cell & $61\%$ & $2.6\%$ \\
            & 7-cell & $48\%$ & $5.2\%$ \\
            \midrule
            \multirow{2}{*}{RCGA} & 3-cell &  $100\%$ & $0.8\%$ \\
            & 7-cell &  $100\%$ & 1$.2\%$ \\
            \midrule
            \bottomrule
        \end{tabular}
    \end{center}
\end{table}

\subsection{Overall Performance and Analysis}

Figs. \ref{fig:1} and \ref{fig:2} show the performance results for various numbers of contents and requests, respectively. Overall, the numerical results demonstrate that RCGA can provide very good solutions with at most 2.4\% deviation from the optimum/LB, and outperforms PBA significantly. 

The results in Fig. \ref{fig:1} present that the cost obtained by RCGA changes only slightly when increasing the number of contents, whereas there is a significant increase for PBA. When the number of contents increases, more relatively unpopular contents will be requested due to the binomial distribution, and these requests for result in a higher download cost for PBA. The solutions offered by RCGA are based on specific requests, so the increasing number of contents hardly impacts on RCGA when the number of requests is fixed. 

Fig. \ref{fig:2} shows the impact of the number of requests on cost. We can see that the costs obtained by RCGA and PBA both increase almost linearly with the number of requests. Thus the relative gap of the solutions offered by both PBA and RCGA from the optimum/LB actually decreases. The reason is that with larger number of requests, the cached contents are also requested more often. 

\subsection{Impact of MCRs}

We are curious about the impact of the proportion of MCRs on RCGA in the aspects of the solution time and the optimality gap from the LB, thus we examine RCGA for the 7-cell instances with different values of $\rho^{\rm m}$ and $\rho^{\rm tt}$. The results are shown in Table \ref{tab:MCRs}. We can see that RCGA can efficiently provide good solutions that are at most $1.59\%$ from the LB. Also, the results in Table \ref{tab:MCRs} demonstrate that the solution time of RCGA is hardly impacted by the number of MCRs. This is because the proportion of MCRs affects the scale of RMP but not the SP, and the time of solving SPs dominates the overall solution time in RCGA.

\begin{table}[h]
    \caption{\label{tab:MCRs} Impact of the proportion of MCRs}
    \begin{center}
        \begin{threeparttable}[b]
            \begin{tabular}{*{4}{c}}
                \toprule
                \midrule
                \quad $\boldsymbol{\rho}^{\rm\bf m}$ \quad \quad&  $\boldsymbol{\rho}^{\rm\bf tt}$\quad & \quad{\bf Time (s)} \quad&{\bf Gap} \quad\\
                \midrule
                \multirow{3}{*}{$20\%$} &  $3:1$  & $4.82$ & $1.20\%$ \\
                           &  $1:1$ & $3.23$ & $0.97\%$ \\
                           &  $1:3$ & $5.38$ & $1.59\%$\\
                \midrule
                \multirow{3}{*}{$40\%$}& $3:1$ & $6.90$ & $0.33\%$ \\
                           &  $1:1$ & $7.22$ & $1.38\%$ \\
                           &  $1:3$ & $3.36$ & $0.51\%$\\
                \midrule
                \multirow{3}{*}{$60\%$} &  $3:1$ & $4.36$ & $1.23\%$ \\
                           &  $1:1$ & $7.36$ & $0.43\%$ \\
                           &  $1:3$ & $8.36$ & $0.52\%$\\
                \midrule
                \bottomrule
            \end{tabular}
            
            \begin{tablenotes}
            	\footnotesize
            	\item[$\dagger$] $R = 5000$, $T = 12$, and $G = 0.3 \sum_{i\in \mathcal{I}} s_{i}$.
            \end{tablenotes}
        \end{threeparttable}
    \end{center}
\end{table}

\subsection{Impact of Backhaul Capacity}

\pgfplotsset{width=.45\textwidth, height=4cm, compat=1.9}

\begin{figure}[t]
    \begin{tikzpicture}
        \centering
        \begin{axis}[
        		xtick style={draw=none}, ytick pos=right, 
                ybar, axis on top,
                tick label style={font=\scriptsize},
                xtick align = inside,
        		ylabel style={align=center},
                ylabel={Percentage (\%)},
                ymin = 0, ymax = 100,
                symbolic x coords={0.1, 0.15, 0.2, 0.25, 0.3},
                xticklabels=\empty,
                enlarge x limits = 0.15,
               	legend style = {at = {(0.98, 0.98)},
        	    legend pos = north east, 
            	legend columns = 1},
        		legend image code/.code={\draw [#1] (0cm,-0.1cm) rectangle (0.2cm,0.15cm); },
        		bar width=.016\textwidth,
        	    tick label style={font=\scriptsize},
        	    label style={font=\small},
        	    legend style={font=\scriptsize},
                minor x tick num=1,
                xminorgrids,
                minor tick length=0,
                grid style={dashed},
        ]
            \addplot  coordinates { (0.1, 55) (0.15, 45) (0.2, 36) (0.25, 34) (0.3, 30) };\label{bar1}
            \addplot  coordinates { (0.1, 23) (0.15, 30) (0.2, 33) (0.25, 37) (0.3, 42) };\label{bar2}
            \addplot  coordinates { (0.1, 20) (0.15, 25) (0.2, 31) (0.25, 29) (0.3, 28) };\label{bar3}
        	\legend{Download cost, Update cost, AoI cost}
        \end{axis}
        
		\begin{axis}[
    		xtick style={draw=none}, ytick pos=left, 
			scaled y ticks=base 10:-5,
		    ylabel={Cost (PBA)},
		    xticklabels=\empty,
		    ymin = 0, ymax=460000,
		    ytick = {0, 100000, 200000, 300000, 400000},
            enlarge x limits = 0.15,
		    legend pos=north west,
		    grid style=densely dashed,
		    tick label style={font=\scriptsize},
		    label style={font=\small},
		    legend style={font=\scriptsize},
		]
    		\addplot[ color=black, mark=o, line width=0.8pt]     
    		coordinates { (0.1, 340000) (0.15, 240000) (0.2, 180000) (0.25, 160000) (0.3, 150000)}; \addlegendentry{Cost}
		\end{axis}
    \end{tikzpicture}
    \begin{tikzpicture}
        \centering
        \begin{axis}[
        		xtick style={draw=none}, ytick pos=right, 
                ybar, axis on top,
                tick label style={font=\scriptsize},
                xtick align=inside,
                xlabel={The scaling factor of backhaul capacity $\rho^b$},
        		ylabel style={align=center},
                ylabel={Percentage (\%)},
                ymin = 0, ymax = 100,
                symbolic x coords={0.1, 0.15, 0.2, 0.25, 0.3},
                enlarge x limits = 0.15,
              	legend style = {at = {(0.98, 0.98)},
        	    legend pos = north east,
        	    anchor = north west, 
            	legend columns = -1},
            	legend image code/.code={\draw [#1] (0cm,-0.1cm) rectangle (0.2cm,0.15cm); },
            	bar width=.016\textwidth,
                tick label style={font=\scriptsize},
                label style={font=\small},
                legend style={font=\scriptsize},
                minor x tick num=1,
                xminorgrids,
                minor tick length=0,
                grid style={dashed},
        ]
            \addplot  coordinates { (0.1, 18) (0.15, 8) (0.2, 3) (0.25, 0) (0.3, 0) };
            \addplot  coordinates { (0.1, 42) (0.15, 54) (0.2, 66) (0.25, 80) (0.3, 84) };
            \addplot  coordinates { (0.1, 40) (0.15, 38) (0.2, 31) (0.25, 20) (0.3, 16) };
        \end{axis}
		\begin{axis}[
    		xtick style={draw=none}, ytick pos=left, 
			scaled y ticks=base 10:-5,
		    ylabel={Cost (RCGA)},
		    xticklabels=\empty,
		    ymin = 0, ymax=460000,
		    ytick = {0, 100000, 200000, 300000, 400000},
            enlarge x limits = 0.15,
		    legend pos=north east,
		    grid style=densely dashed,
		    tick label style={font=\scriptsize},
		    label style={font=\small},
		    legend style={font=\scriptsize},
		]
    		\addplot[ color=black, mark=o, line width=0.8pt]     
    		coordinates { (0.1, 180000) (0.15, 100000) (0.2, 58000) (0.25, 51000) (0.3, 50000)};
		\end{axis}
    \end{tikzpicture}
    \caption{For various backhaul capacity levels in a 7-cell instance, the line shows the cost (labeled on the left y-axis), and the bars show the proportions of three types of cost (labeled on the right y-axis) within the overall cost. The parameters of the instance are $I = 200$, $R = 5000$, $T = 12$, $\rho^{\rm m} = 40\%$, and $\rho^{\rm tt} = 1:1$.} \label{fig:3}
\end{figure}
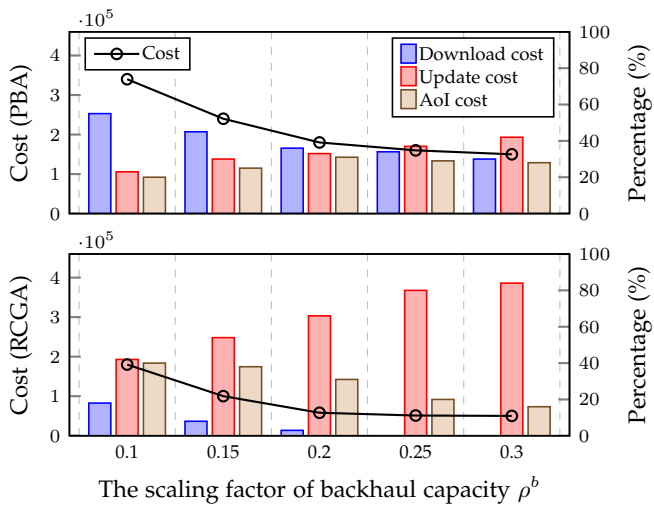

In some cases, backhaul capacity will be the bottleneck resulting in a higher download cost. We evaluate a 7-cell instance under varying backhaul capacity levels. Fig. \ref{fig:3} shows both the cost and the proportions of three types of cost within the overall cost. Apparently, the cost decreases by increasing the backhaul capacity. Moreover, the proportion of download cost also decreases because more requests can be satisfied with the cached and low-AoI contents, instead of using the cloud server. There is a saturation effect on PBA when $\rho^b \geq 0.2$, however. We can see that, for PBA, the download cost always exists no matter how large the backhaul capacity is. This is because those requests for relatively unpopular contents have to be satisfied by the cloud server. For RCGA, the download cost can be decreased to zero when $\rho^b \geq 0.25$, as RCGA is request-tailored. In addition, with larger backhaul capacity, intuitively the proportion of update cost will grow at optimum, whereas in addition to the download cost, the AoI cost proportion shall go down due to higher information freshness. This is indeed the case for RCGA.


\section{Conclusion}

We have considered a content scheduling problem for multi-cell caching where the multiple cells are coupled with each other due to overlapping coverage areas. We have proved the NP-hardness of the problem. We formulated the problem as an integer liner program (ILP). The ILP provides optimal solutions, but it is not practical to large-scale instances in terms of running time. Making the use of the structure of this optimization problem, we have proposed an algorithm based on column generation. The algorithm can efficiently obtain close-optimal solutions that are within a few percent from global optimality. These observations together represent a step closer toward deploying large-scale MEC systems.



\bibliographystyle{IEEEtran}
\bibliography{mybibtex}

%
%

%
%
%

\end{document}